\documentclass{article}

\newcommand{\ipeFigSmallTheta}{1}
\newcommand{\ipeFigLoop}{2}
\newcommand{\ipeFigForward}{3}
\newcommand{\ipeFigSide}{4}
\newcommand{\ipeFigPositiveRouting}{5}
\newcommand{\ipeFigNegativeRouting}{6}
\newcommand{\ipeFigMemorylessNeg}{7}
\newcommand{\ipeFigOneMemoryNeg}{8}
\newcommand{\ipeFigLemmaForward}{9}
\newcommand{\ipeFigLemmaSide}{10}
\newcommand{\ipeFigCurves}{11}

\newcommand{\ipeFigProofTh}{15}

\usepackage{xcolor}
\usepackage{fullpage}

\usepackage[T1]{fontenc} 
\usepackage{hyperref}

\usepackage[pdftex]{graphicx}
\usepackage{amsmath,amsthm,amssymb,enumerate}
\usepackage{yhmath}
\usepackage{cite}
\usepackage[utf8]{inputenc}
\usepackage[OT1]{fontenc}

\usepackage{wasysym}

\newtheorem{theorem}{Theorem}
\newtheorem{lemma}[theorem]{Lemma}

\newtheorem{remark}[theorem]{Remark}

\usepackage{bbold}
\usepackage{xspace}

\DeclareMathOperator{\Probability}{\mathbb{P}}

\DeclareMathOperator{\Areaa}{Area}
\DeclareMathOperator{\Expected}{\mathbb{E}}
\newcommand{\Ex}[1]{\Expected\pbrcx{#1}}
\newcommand{\Prob}[1]{\Probability\pbrcx{#1}}

\newcommand{\Area}[1]{\Areaa\pth{#1}}

\newcommand{\R}{\ensuremath{\mathbb{R}}\xspace}
\newcommand{\N}{\ensuremath{\mathbb{N}}\xspace}

\newcommand{\dd}{{\rm d}}

\newcommand{\indicator}[1]{\mathbb{1}_{[#1]}}
\newcommand{\pth}[1]{\ensuremath{\left(#1\right)}}
\newcommand{\absv}[1]{\ensuremath{\left|#1\right|}}
\newcommand{\pbrcx}[1]{\ensuremath{\left[#1\right]}}
\def\TDdel{{\textit{TD}-Delaunay triangulation}\xspace}

\newcommand{\etal}{\emph{et al.}\xspace}
\newcommand{\spr}{\ensuremath{\delta}}
\newcommand{\rr}{\ensuremath{\rho}}

\newcommand{\halftsix}{half-\ensuremath{\Theta_6}-graph\xspace}

\newcommand{\tsix}{\ensuremath{\Theta_6}-graph\xspace}
\newcommand{\power}{\mathcal{P}}

\usepackage{soul}
\definecolor{OliveGreen}{rgb}{0.2,0.5,0.2}
\definecolor{Grey}{rgb}{0.7,0.7,0.7}
\definecolor{cyan}{rgb}{0.0, 0.72, 0.92}
\definecolor{Pink}{rgb}{1,0,1}
\newcommand{\mycomment}[1]{\marginpar{\tiny  #1}}
\newcommand{\od}[1]{{\color{OliveGreen}  #1}}
\newcommand{\olivier}[1]{\mycomment{\od{#1}}}

\newcommand{\sodareview}[1]{\mycomment{{\color{red}SODA review: #1}}}

\renewcommand{\od}[1]{#1}\renewcommand{\olivier}[1]{}\renewcommand{\sodareview}[1]{}

\usepackage[mathlines]{lineno}
\newcommand*\patchAmsMathEnvironmentForLineno[1]{%
  \expandafter\let\csname old#1\expandafter\endcsname\csname #1\endcsname
  \expandafter\let\csname oldend#1\expandafter\endcsname\csname end#1\endcsname
  \renewenvironment{#1}%
     {\linenomath\csname old#1\endcsname}%
     {\csname oldend#1\endcsname\endlinenomath}}%
\newcommand*\patchBothAmsMathEnvironmentsForLineno[1]{%
  \patchAmsMathEnvironmentForLineno{#1}%
  \patchAmsMathEnvironmentForLineno{#1*}}%
\AtBeginDocument{%
\patchBothAmsMathEnvironmentsForLineno{equation}%
\patchBothAmsMathEnvironmentsForLineno{align}%
\patchBothAmsMathEnvironmentsForLineno{flalign}%
\patchBothAmsMathEnvironmentsForLineno{alignat}%
\patchBothAmsMathEnvironmentsForLineno{gather}%
\patchBothAmsMathEnvironmentsForLineno{multline}%
}

\makeatletter
\newcommand\footnoteref[1]{\protected@xdef\@thefnmark{\ref{#1}}\@footnotemark}
\makeatother

\author{
Prosenjit Bose,\thanks{
Carleton University, Ottawa, Canada. 
{\tt jit@scs.carleton.ca}}
\and
Jean-Lou De Carufel,\thanks{University of Ottawa,  Ottawa, Canada. {\tt jdecaruf@uottawa.ca }}
\and
Olivier Devillers,\thanks{
Universit\'e de Lorraine, CNRS, Inria, LORIA, F-54000 Nancy,
France. {\tt olivier.devillers@inria.fr}.}
}

\title{{
Expected Complexity of Routing in $\Theta_6$  and Half-$\Theta_6$ Graphs\thanks{%
\scriptsize
This work has been supported by INRIA Associated team TRIP,
by grant ANR-17-CE40-0017 of the French National Research Agency (ANR project ASPAG) and
by NSERC.
}}}
\date{}

\begin{document}

\maketitle

\begin{abstract}
We study online routing algorithms on the \tsix and the \halftsix (which is equivalent to 
a variant of the Delaunay triangulation). Given a source vertex $s$ and a target vertex $t$ in the \tsix (resp. \halftsix), there
exists a deterministic online routing algorithm that finds a path from
$s$ to $t$ whose length is at most $2\|st\|$ (resp. $2.89\|st\|$)
which is optimal in the worst case [Bose~\etal, {\sc siam} J. on Computing, 44(6)]. We propose alternative, slightly simpler routing algorithms that are optimal in the worst case and for which we provide an analysis of the average routing ratio for the \tsix and \halftsix defined on a Poisson point process. 

For the \tsix,  our online routing algorithm has an expected routing
ratio of $1.161$ (when $s$ and $t$ are random) and a maximum expected routing
ratio of $1.22$
(maximum for fixed $s$ and $t$ where all other points are random),
much
better than the worst-case routing ratio of $2$. 
For the \halftsix, our memoryless online routing algorithm has an
expected routing ratio of $1.43$ and a maximum expected routing ratio
of $1.58$. Our online routing algorithm that uses a constant amount of additional  memory
has an expected routing ratio of $1.34$ and a maximum expected routing
ratio of $1.40$. The additional memory is only used to remember the
coordinates of the   starting point of the route.
Both of these algorithms have an expected routing
ratio that is much better than their worst-case routing ratio of $2.89$. 
  
\end{abstract}


\section{Introduction}

A weighted geometric graph $G=(P,E)$ is a graph whose vertex set is a set $P$
of $n$ points in the plane, and whose edge set is a set of line segments
joining pairs of points in $P$, with each edge weighted by the Euclidean
distance between its endpoints.
A graph $G$ is
a geometric $\spr$-spanner of the complete geometric graph
provided that for every pair of points $(s,t)\in P^2$
the shortest path from $s$ to $t$ in $G$ has weight at most
$\spr\geq 1$ times $\|st\|$, where $\spr$ is the {\em spanning ratio} or {\em
  stretch factor} and $\|st\|$ is the Euclidean distance from $s$ to
$t$.
The spanning properties of various
geometric graphs have been studied extensively in the literature (see
\cite{BS11,NS-GSN-06} for a comprehensive overview of the topic).  

A routing algorithm for a geometric graph takes as input a pair of vertices
$(s,t)$ and finds a path from $s$ to $t$ in the graph.  When full knowledge of
the graph is available to the algorithm, numerous routing algorithms exist in
the literature for finding paths in these graphs such as Breadth-First Search,
Depth-First Search or Dijkstra's algorithm~\cite{m59,l61,ht73,
dijkstra1959note}.  The problem offers different challenges in the {\em
online} setting. By the online setting, we mean that
initially, the routing algorithm has limited knowledge of the graph and needs
to simultaneously explore the graph while trying to find a path from $s$ to
$t$. Without knowledge of the whole graph, a routing algorithm, in general,
cannot identify a short path. In certain cases, depending on the information
available to the routing algorithm and limitations placed on the algorithm such
as how much memory it has, it may not even find a path but may end up cycling
without ever reaching its destination~\cite{BoseM04}. A formal definition 
of our routing model is given in Section \ref{sec:model}.

A graph $G$ being a spanner of the complete graph simply implies the
existence of a short path in $G$ between every pair of vertices. The
goal of a competitive online routing algorithm is to find a short path
when one exists. A routing algorithm is \emph{$\rr$-competitive} if
for any two points $s$ and $t$, the length of the path in $G$ followed
by the routing algorithm is not more than $\rr$ times $\|st\|$, where
$\rr$ refers to the \emph{routing ratio}~\cite{Theta6-routing}. There is an
intimate connection between the spanning and routing ratio. The
routing ratio can be viewed as the spanning ratio of the path found by
the routing algorithm, thus the routing ratio is an upper bound on the
spanning ratio. 

Both the spanning ratio and the routing ratio are fragile
measures. For example, $G$ can be the complete graph that is missing
only one edge but can still have an unbounded spanning and routing
ratio. As such, studying the spanning and routing ratio in the
expected sense is a more robust measure. One of the key difficulties
in analyzing these ratios in the probabilistic sense is that often
there is a lot of dependence in the process used by a routing
algorithm to select which edge to follow. To overcome these barriers,
we need to define simple routing strategies with good worst-case
behaviour that can also be analyzed in the expected sense.

\subsection{Contribution}
This paper has two main contributions.
The first contribution consists of the design of two 
new algorithms for routing in the \halftsix
  (also known as the $TD$-Delaunay  triangulation~\cite{BGHI10})  
 in the so called {\em negative-routing} case, which is challenging
 since at each step, the routing algorithm must select among many
 possible edges to follow, some of which lead you astray. Our new routing algorithms come in 
 two flavors: one is memoryless and the other uses a constant amount of memory.
These new negative-routing algorithms have a worst-case optimal routing ratio
but are simpler and more amenable to probabilistic analysis than the known optimal routing algorithm~\cite{Theta6-routing}.  
We also provide a new point of view on routing~\cite{Theta6-routing} in the \halftsix in the
{\em positive-routing} case, which in some sense is identical to the optimal routing algorithm on the \tsix. This new point of view
allows us to complete the probabilistic analysis both on the \halftsix and on the \tsix.

The second contribution is the analysis of the two new negative-routing algorithms
and of the positive-routing algorithm in a random setting, namely when the vertex set of the \tsix and \halftsix is a point
set that comes from an infinite Poisson point process $X$ of intensity $\lambda$. 
The analysis is asymptotic with $\lambda$ going to infinity,
and gives the expected length of the shortest path between two fixed points $s$ and
$t$ at distance one. Our results depend on the position of $t$ with
respect to $s$. We express our results both by taking the worst
position for $t$ and by averaging over all possible positions for $t$.

The routing ratio for our memoryless negative-routing algorithm in the
\halftsix is
2.89 in the worst case which is optimal, 1.58 in the expected case for
the worst position for $t$, and 1.43 in the expected 
case when averaging over all possible positions for $t$.
For our constant memory negative-routing algorithm, we obtain a
routing ratio of
2.89 in the worst case, 1.40 expected for the worst position for $t$,
and 1.34 averaging on all possible positions for $t$. 
For the routing ratio of the positive-routing algorithm on the \halftsix, we obtain
2 in the worst case which is optimal, 1.22 expected for the worst
position of $t$, and 1.16 averaging on all possible positions for
$t$. 
Our results on routing in the \tsix are identical to the
positive-routing strategy since the \tsix is the union of two   
half-$\Theta_6$-graphs and one can locally differentiate the edges 
between the two spanning subgraphs.
Therefore, this algorithm for $\Theta_6$-routing is memoryless.

 Formal definitions of 
 the online routing model and the different graphs on which we route 
 are outlined in the following subsections. 

\subsection{The Poisson Point Process} \label{sec:poisson}

A Poisson point process $X$ of intensity $\lambda$ in the plane is an 
infinite set of points satisfying the following
properties:
the expected number of points of $X$ in a domain $A$ is $\lambda\cdot \Area{A}$
and the number of points of $X$ in two disjoint domains are
independent.
The number of points in $A$ follows a Poisson's law:
\[
   \Prob{|X\cap A|=k} = \frac{\lambda^k\Area{A}^k}{k!} e^{-\lambda\Area{A}}
  \]

\subsection{The Online Routing Model} \label{sec:model}

In its weakest form, an online local geometric routing algorithm
on a graph $G=(P,E)$
can be
expressed as a {\em routing function}
$f:P\times P \times\power(P) \to P$, where $\power(\cdot)$ denotes the power set, with
parameters $f(u,t,N(u))$ such that $u \in P$ is the vertex for which a
forwarding decision is being made (i.e., the node currently holding the
message), $t \in P$ is the destination vertex (target), and $N(u) \subseteq P$ is
the set of neighbours of $u$ in $G$.  Upon receiving a message
destined for $t$, node $u$ forwards the message to its neighbour
$z =f(u,t,N(u))\in N(u)$. This routing strategy is referred to as {\em memoryless} routing.
If the routing algorithm uses constant additional
memory
{to store some information $i\in\mathcal{I}$, then this information is}
taken into consideration when computing
which neighbour to forward the message to. The function then becomes
$f:P\times P \times \power(P) \times {\mathcal{I}} \to P$. Such a routing
strategy is referred to as {\em constant-memory routing.} In the remainder of
the article, we use the additional memory to store 
the vertex coordinates of the source of the message.

\subsection{$\Theta_k$-routing\label{s:thetakrouting}}
For any $k$ the $\Theta_k$-graph is defined as follows. For each point
{$p\in P$}, 
consider a set of rays originating from~$p$ with the angle between 
consecutive rays being $2\pi / k$. Each consecutive pair of rays 
defines a cone. Orient the cones such that there is one cone, labeled 
$C^p_0$, whose bisector is a vertical ray through $p$ pointing 
upwards. Label the cones in counterclockwise 
order: $C^p_0, \ldots, 
C^p_{k-1}$.
 Given two vertices $p$ and $q$ define the 
\emph{canonical triangle} $T_{pq}$ to be the triangle bounded by the
sides of the 
cone of $p$ that contains $q$ and the line through $q$ perpendicular 
to the bisector of that cone. An edge in $\Theta_k$ exists between two 
vertices $p$ and $q$ if $q$ is in some cone $C^p_i$, and for all 
points $w \in C^p_i$, $\|pq'\| \leq \|pw'\|$, where $p'$ and $w'$ denote 
the orthogonal projection of $p$ and $w$ onto the bisector of cone 
$C^p_i$.
In other words, $T_{pq}$ {contains no points of the point set $P$.
  We say that $T_{pq}$ is {\em empty}.}
The half-$\Theta_k$-graph is 
defined similarly for even $k$ but only half the cones are considered 
for edge inclusion.  Thus, an edge exists between two vertices $p$ and 
$q$ of the half-$\Theta_k$-graph provided that $q$ is in some cone 
$C^p_i$ where $i$ is even, and $T_{pq}$ is empty.
The {\em even cones} refer to the cones with even index 
and the {\em odd cones} refer to the ones with odd index. In fact, the 
$\Theta_k$-graph is the union of the half-$\Theta_k$-graph defined by 
the even cones and the half-$\Theta_k$-graph defined by the odd cones.

The structure of the $\Theta_k$-graph naturally gives rise to a simple 
routing algorithm known as \emph{$\Theta_k$-routing}. Let $t$ be the 
destination vertex. The $\Theta_k$-routing algorithm invoked at an 
%
%
arbitrary vertex $v$ consists of following the edge adjacent to $v$ in 
the cone of $v$ that contains $t$. This process is repeated until the 
destination $t$ is reached. 

It is known that $\Theta_k$-routing 
terminates with routing ratio $\rr = 1+f(k)$ where $f(k) \in o(1)$ for 
all $\Theta_k$-graphs with $k\geq 7$
\cite{ruppert1991approximating, BCMRV16}.
There is a gap between the best known upper bound
on the spanning ratio and $\rr$  (see \cite{BCMRV16} for a survey of 
the best known bounds both on the spanning ratio and $\rr$).
\begin{figure}[t]
     \begin{center}
       \includegraphics[page=\ipeFigSmallTheta,width=0.8\textwidth]{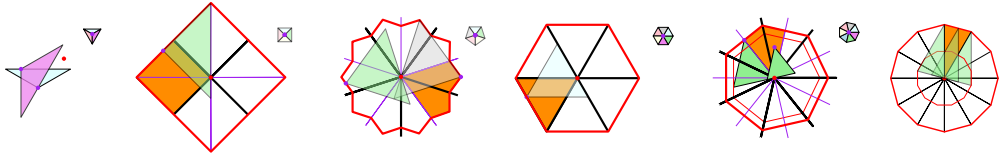}
     \end{center}
     \caption{
       {Routing in the $\Theta_k$-graph for $k\in\{3,4,5,6,7,8\}$
           (from left to right). 
        $\Theta_3$-routing to the red point loop between the two 
       purple points. $\Theta_k$-routing ($k\in\{4,5,6\}$) 
       from the red polygon to the red point goes strictly inside the 
       polygon, thus decrease the polygonal distance to the target and 
       prove that routing will terminates on a finite set of 
       points. 
       For $k\in\{7,8,\ldots\}$, one step from the boundary of the red 
       polygon allows to quantify the decrease in polygonal distance 
       to the target and prove that the routing ratio is bounded. 
     }
        \label{fig:small-theta}
     }
 \end{figure}
For $k=2$, the graph is just the $y$-monotone chain of vertices 
ordered vertically. In this case, $\Theta_2$-routing works but the routing ratio is unbounded. 
For $k=3$, the graph is connected but $\Theta_3$-routing may loop 
as on the example of Figure~\ref{fig:small-theta}-left~\cite{ABBBKRTV14}. 
For $k \in \{4,5,6\}$, 
$\Theta_k$-routing always finds a path but its length may be unbounded 
(see Figures~\ref{fig:small-theta} and~\ref{fig:loop}). 
Alternative routing algorithms dedicated specifically for $\Theta_4$, $\Theta_5$, and $\Theta_6$ graphs 
have been designed proving that they have constant routing ratio respectively 
smaller than 17~\cite{Theta4}, $\sqrt{50+22\sqrt{5}}<10$~\cite{Theta5}, 
and 2~\cite{BGHI10,Theta6-routing}. 
\begin{figure}[t]
     \begin{center}
         \includegraphics[page=\ipeFigLoop,width=0.8\textwidth]{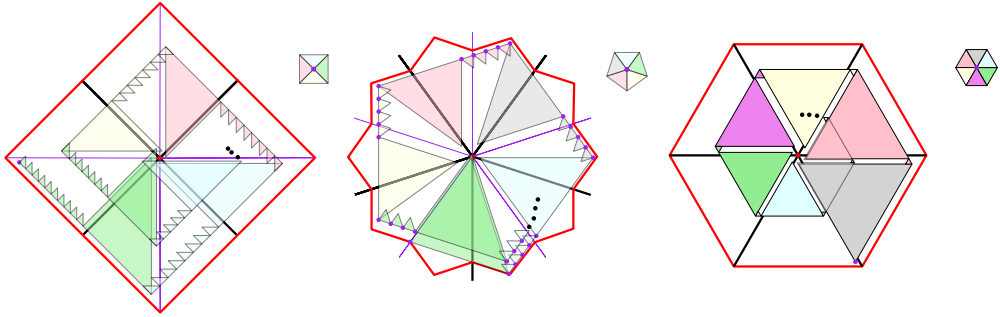}
     \end{center}
     \caption{
Routing ratio in the $\Theta_k$-graph is unbounded for $k\in\{4,5,6\}$
        \label{fig:loop}
     }
 \end{figure}
 Bonichon and Marckert
\cite{bonichon_marckert_2011} 
analyze the expected length of the $\Theta_k$-routing algorithm when 
the $\Theta_k$-graph is defined on a Poisson point process 
with intensity $\lambda$ tending to infinity.
Their results are much more complete since they address all $k$ and
variants of $\Theta_k$-graphs such as Yao-graphs or continuous
$\Theta_k$-graphs.
{But these results are limited to the standard $\Theta_k$-routing algorithm
  while we will address different routing algorithms in the sequel.}

\subsection{The \tsix and Half-\tsix}
In the special case $k=6$, the \tsix has some interesting properties.
The six rays that define the cones around a point $p$ make an angle of 
$0$, $\frac\pi3$, $\frac{2\pi}3$, $\pi$, $\frac{4\pi}3$, and $\frac{5\pi}3$ 
with the horizontal axis. 
The triangles $T_{pq}$ when  $q$ is in an even (resp. odd) cone $C_i^p$ are all
homothets. This is the key property that is not shared with any other
$\Theta$-graphs.
For other values of $k$ cones are 
homothets only when the value of $i$ is fixed. For $k=6$ only the
parity of $i$ matters.
Using this property,
Bonichon~\etal~\cite{BGHI10} noted that the \halftsix is equivalent 
to the \TDdel.
The \TDdel is a variant of the standard Delaunay triangulation 
where the empty disk property is replaced with an empty homothet of an equilateral triangle 
(\textit{TD} is an abbreviation for triangular distance).

\subsection{$\Theta_6$-routing\label{s:theta6routing}}

  As mentioned in Section~\ref{s:thetakrouting}, $\Theta_6$-routing
  always terminates but may have an unbounded routing-ratio
(see Figures~\ref{fig:small-theta} and~\ref{fig:loop}). 
The analysis by Bonichon and Marckert 
\cite{bonichon_marckert_2011}
on $\Theta_k$-routing for $k=6$ bounds the expected
cost of the $\Theta_6$-routing algorithm 
when the point set is defined on a Poisson point process 
with intensity $\lambda$ tending to infinity. 
As noted in Remark~\ref{r:theta6routing2}, our techniques establish the same 
probabilistic bounds. 
However the worst case routing ration of the $\Theta_6$-routing
algorithm is unbounded, see Figure~\ref{fig:loop}.
We focus on the probabilistic analysis of 
routing algorithms that are optimal in the worst-case.

Chew~\cite{Chew89} showed that the \TDdel (equivalently the 
\halftsix) is a 2-spanner. His proof is constructive, however, it does 
not provide an online routing algorithm that successfully routes 
between every ordered pair of vertices. Without loss of generality,
label the cones of the \halftsix such that the 
\TDdel is equivalent to the even \halftsix. In this case,
positive-routing refers to routing from $s$ to $t$ when $T_{st}$ is
even and negative-routing refers to routing from $s$ to $t$ when
$T_{st}$ is odd.  
Chew's algorithm is a positive-routing algorithm and constructs a path from $s$ to $t$ with a routing ratio of 2 when $T_{st}$ is even. Since for every pair of points $s$ and $t$, either $T_{st}$ is even or $T_{ts}$ is even, Chew's algorithm proves 
that the \TDdel is a geometric 2-spanner. However, when $T_{st}$ is odd, Chew's routing algorithm fails. 

Bose~\etal~\cite{Theta6-routing} addressed the negative-routing case  
providing an algorithm with routing ratio 
$\frac{5}{\sqrt{3}}\simeq 2.89$. 
Surprisingly, this ratio is optimal for any constant-memory online routing algorithm~\cite{Theta6-routing}. 
This algorithm and the  one for the positive case are detailed in Sections~\ref{s:positive} and~\ref{s:negative}. 
Since the worst-case optimal routing
ratio is $2.89$ and worst-case optimal spanning ratio is $2$, this  is
one of the rare known separations between the spanning ratio and  
routing ratio of a spanner in the online
setting~\cite{Theta6-routing}.
In essence, even though there exists a path from $s$ to $t$
whose length is at most $2\|st\|$, no constant memory routing
algorithm can find this path in the worst-case.

 \section{Two Basic Routing Building Blocks on the Half-\tsix \label{s:basic}}

   We introduce two routing modes on the \halftsix which serve as
   building blocks for our routing algorithms that have optimal
   worst-case behaviour. We consider the even \halftsix 
   and for ease of reference, we color code the cones $C_0, C_2$ and $C_4$
    blue, red, and green 
   respectively. The two building blocks are: the {\em forward-routing phase} and the {\em side-routing phase}. 
{
Then these two building blocks are used to reformulate the routing algorithms
proposed by Bose at al.~\cite{Theta6-routing}.
}

\subsection{Forward-Routing Phase}

Forward-routing consists of only following edges defined by a specific type of cone (i.e., a cone with the same color) until some specified stopping condition is met. 
For example, suppose the specific cone selected for forward routing is the blue cone. Thus, when forward-routing
is invoked at a vertex $x$, the edge followed is $xy$ where $y$ is the
vertex
\begin{figure}[t]     \begin{center}
         \includegraphics[page=\ipeFigForward,width=\textwidth]{./Figures}
     \end{center}     \caption{
       A \TDdel (\halftsix).
        \label{fig:forward}     } \end{figure}
adjacent to $x$ in $x$'s blue cone. If  
the stopping  
 condition is not met at $y$, then the next edge followed is $yz$ where $z$ is the vertex adjacent to $y$ in $y$'s blue cone. This process continues until a specified stopping condition is met.  
A path produced by forward routing consists of edges of the same color since edges are selected from one specific cone
as  illustrated in Figure~\ref{fig:forward}.

\begin{lemma}\label{f:forward-wc-length}
Suppose that forward-routing is invoked at a vertex $s$ and ends at a vertex $t$.
The length of the path from $s$ to $t$
produced by forward-routing is at most the length of one side of the canonical triangle $T_{st}$ which is
$\frac{2}{\sqrt{3}}$ times 
the length of the orthogonal projection of $st$ onto the bisector of $C_0^s$.
\end{lemma}

  \begin{proof}
This result follows from the fact that each edge along the path makes a maximum angle of $\frac\pi6$ with the
cone bisector and the path is monotone in the direction of the cone bisector.
  \end{proof}

\subsection{Side-Routing Phase in the Half-\tsix}

The {\em side-routing phase} is defined on the \halftsix by using the
fact that it is 
the \TDdel, and thus planar.
Consider a line $\ell$ {parallel to one of the cone sides.
Without loss of generality, we will assume $\ell$ is horizontal.
We call the side of the line that bounds the even cones the {\em positive
  side} of $\ell$. For a horizontal line, the positive side is below
$\ell$, and for the lines with slopes $-\sqrt3$ and $\sqrt3$,
respectively, the positive side is above the line.}
Let
$\Delta_1,\ \Delta_2,\ \Delta_3,\ldots$ be an ordered sequence
of consecutive triangles of the \TDdel
intersecting $\ell$.
Let $j\in \N^\star$ and let $B$ be the  
{piece of the} boundary of the union of the triangles $\Delta_1,\ldots, \Delta_j$ that
{goes from $s$, the bottom-left vertex of $\Delta_1$,
  to $t$, the bottom-right vertex of $\Delta_j$, below $\ell$.}
Note that $B$ is a path 
in the \halftsix. 
Side-routing invoked at vertex $s$ along $\ell$ stopping at $t$
consists of walking from $s$ to $t$
along $B$ (see Figure~\ref{fig:side} for an example).

\begin{lemma}\label{f:side-wc-length}
  Side-routing on the positive side of a line $\ell$ {parallel to a
    cone boundary} invoked at a
  vertex $s$  and stopped at a vertex $t$
  in the \halftsix results in a path whose length is bounded 
  by twice the length of the orthogonal projection of $st$ on $\ell$.
{This path only uses edges of two colors and all vertices of the path
  have their successor of the third color on the other side of $\ell$.}
\end{lemma}

\begin{proof}
Without loss of generality, assume $\ell$ is horizontal and the positive side is below $\ell$.
  Consider the triangles $\Delta_i$, $1\leq i\leq j$ as defined above.
    The empty equilateral triangle $\nabla_i$ circumbscribing $\Delta_i$
    has a vertex of $\Delta_i$ on each of its sides by construction
    (the $\nabla_i$  are shown in grey in Figure~\ref{fig:side}).
    If $\Delta_i$ has an edge of the path (i.e., below $\ell$)
    then the vertex on the horizontal side of $\nabla_1$ is above the line while the two
    others are below.
Thus, such an edge of the path goes from the left to the right
  side of $\nabla_i$. Based on the slopes of the edges of $\nabla_i$, we have the following:
 \\ --a-- Each edge on the path has a length smaller than twice its   horizontal projection.
 Therefore, summing the lengths of all the projections of the edges gives 
  the claimed bound on the length.
\\  --b-- If the slope is negative, the path edge is green and if the slope is positive, the path edge is red.
\\ --c-- The blue successor of a vertex $u$ of $\Delta_i$ on the lower
  sides of $\nabla_i$ is above $\ell$ since the part of $C_o^u$ below
  $\ell$ is inside $\nabla_i$ and thus contains no other points.
  Therefore, blue edges do not appear on the path. 
\end{proof}
\begin{figure}[t]
     \begin{center}
         \includegraphics[page=\ipeFigSide,width=0.8\textwidth]{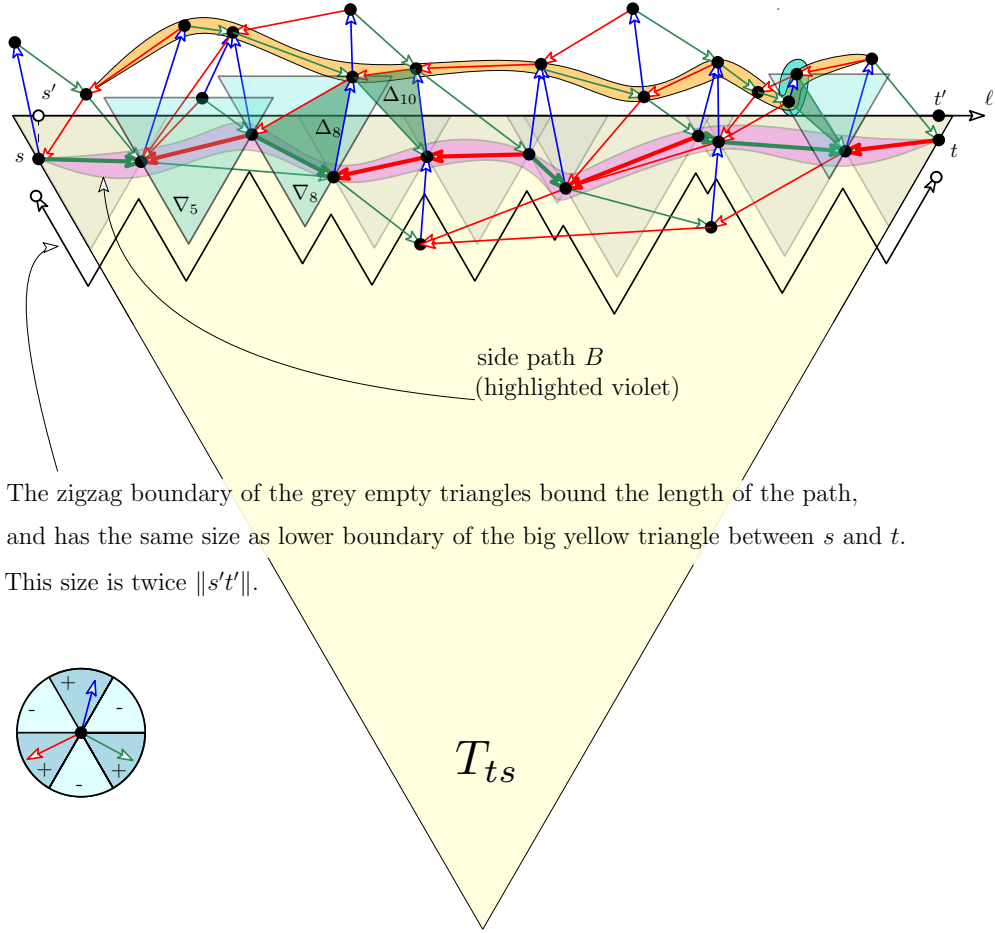}
     \end{center}
     \caption{
A side path below the horizontal line $\ell$. 
        \label{fig:side}
     }
 \end{figure}

{Notice that on the negative side of the line, we do not have the same
properties.} For example, the path above 
$\ell$ in Figure~\ref{fig:side}  uses at least one edge of each color 
(the path above is highlighted in orange, 
note there 
is one blue edge circled in blue). 

\subsection{Positive in the Half-\tsix (and the \tsix)\label{s:positive}}

If $t$ is in a positive cone of $s$,
Bose et al.~\cite{Theta6-routing}
proposed a routing algorithm in the \halftsix which they called \emph{positive routing}.
This algorithm consists of two phases: a \emph{forward-routing} phase and a \emph{side-routing} phase.
The forward-routing phase is invoked with source $s$ and destination
$t$. It produces a path from $s$ to the first vertex $u$ outside the
negative cone of $t$ that contains $s$.
The side-routing phase, invoked with source $u$ and destination $t$, finds a path along the boundary of this negative cone.

For completeness,
we give a proof of the following lemma shown in~\cite{Theta6-routing}.
\begin{lemma}
  \label{l:positive-wc-length}
Positive routing has a worst-case routing ratio of  $2$. 
\end{lemma}
\begin{proof}
Without loss of generality, we assume $t$ is in the positive cone
{$C_0^s$} and  that the 
\begin{figure}[t]
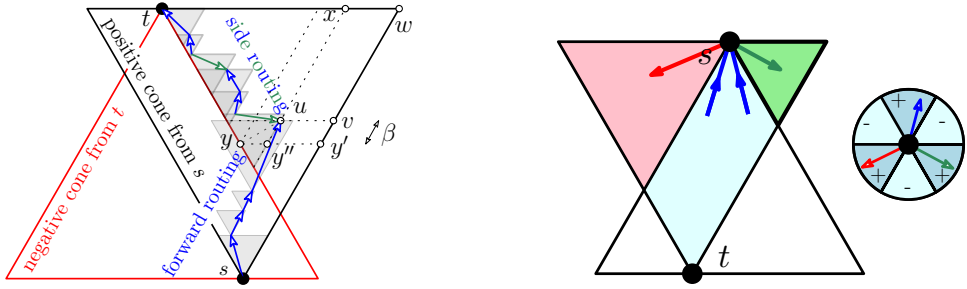
     \begin{center}
         \includegraphics[page=\ipeFigPositiveRouting,width=0.33\textwidth]{./Figures}\vspace*{2mm}
         \hspace*{0.1\textwidth}
         \includegraphics[page=\ipeFigNegativeRouting,width=0.33\textwidth]{./Figures}
     \end{center}
     \caption{
Positive and negative routing schemes~\cite{Theta6-routing}. 
        \label{fig:routing}
     }
 \end{figure}
forward routing leaves the negative cone from $t$ through its  right side.
Let $u$ be the last vertex on the forward routing path,
$x$, $v$ the projections of $u$ on $\partial T_{st}$
(the boundary of $T_{st}$)
parallel to
  its sides; $y$
 the perpendicular projection of $u$ on $\partial T_{ts}$
 ; and $y'$ and $y''$ horizontal projections of $y$ on
 $\partial T_{st}$ and line $xu$
 (see Figure~\ref{fig:routing}-left).
  By Lemma~\ref{f:forward-wc-length}, the length of the path from $s$ to $u$ is
bounded by  $\|sv\|=\|sy'\|+\|y'v\|$ and 
by Lemma~\ref{f:side-wc-length}, the length of the path from $u$ to $t$ is
bounded  by  $2\|yt\|=\|xy''\|+(\|tx\|-\|y''y\|)$. 
Since the triangle $yy''u$ is isosceles we have 
$\|yy''\|=\|y''u\|=\|y'v\|=:\beta$. 
Combining the two paths, the total length is bounded 
 by \[ \|sy'\|+\beta+\|y''x\|+\|tx\|-\beta
=\|sy'\|+\|y'w\|+\|xt\|
\leq\|sw\|+\|wt\|.\]
 Thus, the 
 stretch factor is smaller than $\frac{\|sw\|+\|wt\|}{\|st\|}$.
{Studying the function $\xi\leadsto \frac{1+\xi}{\sqrt{\frac34+\pth{\xi-\frac12}^2}}$
and its derivative, this stretch factor is}
maximal and equal to $2$ when {$\|wt\|=\|sw\|$ ($\xi=\frac{\|wt\|}{\|sw\|}=1$) 
$t$  in the upper left corner in  Figure~\ref{fig:routing}-left}. 
\end{proof}

\subsection{Negative Routing in the Half-\tsix \label{s:negative}}
 When $t$ is in a negative cone of $s$,
Bose et al.~\cite{Theta6-routing} propose a routing strategy which
they call {\em negative routing}. Without
loss of generality, assume that $t \in C^s_3$, which implies that
$s\in C^t_0$ (i.e., the blue cone of $t$ in Figure~\ref{fig:routing}-right). Notice that $T_{ts}$ is
partitioned by $T_{st}$
into three pieces, the portion contained in $C^s_2$ (red  cone of $s$), $C^s_3$ and $C^s_4$ (green cone of $s$).
We refer to these three zones as: the red triangle, the blue region (which is the intersection of the blue cone of $t$ with $T_{st}$) and the green triangle.
Without loss of generality, we assume that
$s$ is to the right of $t$ and the green triangle is smaller than the red one.

Bose et al.'s algorithm~\cite{Theta6-routing}
can be reformulated in the following way:
If neither the green nor the red triangle is empty,
negative routing follows the edge from $s$ into the smaller of the two triangles.
If one of the green or the red triangle is empty, negative routing uses one step of
side routing along the side of the empty triangle
(if both are empty, choose the larger of the two).

This process is iterated until $t$ is reached.
The worst-case stretch factor of this process is 
$\frac{5}{\sqrt{3}}\simeq 2.89$ \cite{Theta6-routing}.

\section{Alternative Negative Routing Algorithms in the Half-\tsix\label{s:alternative-neg}}
In this section, we outline {two} alternatives to the negative routing algorithm described by Bose et al.~\cite{Theta6-routing}. Our algorithms are a little simpler to describe, have the same worst-case routing ratio, 
and are easier to analyze in the random setting.
The lower bound of
  $\frac{5}{\sqrt{3}}\simeq 2.89$ \cite{Theta6-routing}
  applies to our alternative negative routing algorithms.

\subsection{Memoryless Routing\label{s:mem-less-negative}}

{
\begin{enumerate}
\renewcommand{\theenumi}{Case~\arabic{enumi}}
\item 
  \label{step:forward}
    If $t$ is in the positive cone $C_i^s$ ($i$ even),
    take one step of forward-routing   in the direction of $t$
\item 
    \label{step:side-cw} 
    If $t$ is in the negative cone $C_i^s$ ($i$ odd) and the successor $u$
    of $s$ in $C_{i-1}^s$ is outside $T_{ts}$ (red triangle empty),
    take one step of side-routing  along the side of $T_{ts}$ crossed 
    by $su$. 
\item 
    \label{step:side-ccw} 
    If $t$ is in the negative cone $C_i^s$ ($i$ odd) and the successor $u$
    of $s$ in $C_{i+1}^s$ is outside $T_{ts}$ (green triangle empty), 
    take one step of side-routing  along the side of $T_{ts}$ crossed 
    by $su$. 
\item 
    \label{step:neg-forward} 
    If $t$ is in the negative cone $C_i^s$ ($i$ odd) and both successors 
    of $s$ in $C_{i-1}^s$ and $C_{i+1}^s$ are inside $T_{ts}$ (green 
    and red triangle non empty),
    take one step of forward-routing   in the direction of the side of 
    $T_{ts}$ incident to $t$ closest to $s$ (go to the green successor of $s$). 
\end{enumerate}

Beyond the presentation, our strategy differs from the one of Bose et al.~\cite{Theta6-routing} in 
\ref{step:neg-forward} where Bose et al. follows a blue edge 
if one exists. 

We remark that when we reach \ref{step:side-ccw},
we enter a side-routing phase that will continue until $t$ is reached
since a side-routing step ensures that at the next iteration
side-routing is still applicable.

The same argument holds in \ref{step:side-cw},
unless we reach a point $s$ with both successors outside $T_{ts}$
in which case we follow the other side of $T_{ts}$.

To summarize, if $t$ is in a positive cone of $s$, this routing algorithm
will produce the path described in Section~\ref{s:positive}
and Lemma~\ref{l:positive-wc-length} applies.
If $t$ is in a negative cone of $s$} %
we use a forward phase in the green triangle, until we reach a vertex
$u$ whose edge in the green triangle 
intersects $T_{ts}$ (recall that we assume that the green triangle is the smaller one).
At this point, we invoke side-routing from $u$ to $t$
along the boundary of $T_{ts}$.
\begin{figure}[t]     \begin{center}
         \includegraphics[page=\ipeFigMemorylessNeg,width=0.33\textwidth]{./Figures}
         \hspace*{0.1\textwidth}
         \includegraphics[page=\ipeFigOneMemoryNeg,width=0.33\textwidth]{./Figures}
     \end{center}
     \caption{
For Lemmas~\ref{l:memless-wc} and~\ref{l:mem-one-wc}
        \label{fig:lemmas-wc}
     }
 \end{figure}

 \begin{lemma}\label{l:memless-wc}
Memoryless negative routing has a worst-case routing ratio of 
$\frac{5}{\sqrt{3}}\simeq 2.89$. 
\end{lemma}
\begin{proof}
  Assume without loss of generality $s\in C_0^t$. 
  Referring to Figure~\ref{fig:lemmas-wc}-left,
  let $w$ be the upper right vertex of $T_{ts}$,
  $v$  be the orthogonal projection of $u$ on $tw$ and
$x$ its projection parallel to $tw$ on $sw$.
By Lemma~\ref{f:forward-wc-length}, the path from $s$ to $u$ has length 
bounded by  $\|sx\|$ and 
by Lemma~\ref{f:side-wc-length}, the path from $u$ to $t$ has length 
bounded by  $2\|vt\|$.
Combining the two paths the length is bounded by 
$\|sx\|+2\|vt\|\leq\|sw\|+2\|wt\|$. 
Thus the 
stretch is smaller than $\frac{\|sw\|+2\|wt\|}{\|st\|}$.
{Studying the function $\xi\leadsto \frac{2+\xi}{\sqrt{\frac34+\pth{\xi-\frac12}^2}}$
this stretch factor}
attains its 
maximum value of $\frac5{\sqrt3}$ when $s$ and $t$ lie on a  vertical
line {($\xi=\frac{\|wt\|}{\|sw\|}=\frac12$).}
\end{proof}

\subsection{Constant-Memory Negative Routing\label{s:mem-one-negative}}

We propose a second negative routing algorithm that has the same
worst-case routing ratio, but we will prove that it has a better
average routing ratio. However, it is 
no longer memoryless since it needs to remember the coordinates of one
vertex, namely the source $s$ of the path. 

Let $x''$ be the intersection between $T_{ts}$ and $T_{st}$ closest to $s$. 
(see Figure~\ref{fig:lemmas-wc}-right). 
{The idea is to use} side-routing from $s$ along $sx''$
and, just before exiting the green triangle, apply side-routing
along $x''t$.

This routing algorithm is identical to the one in the previous subsection, except that we replace \ref{step:neg-forward} with the
 following, where $u$ is the current vertex and $s$ is the origin of the
 path whose coordinares are kept in memory:
\begin{enumerate}
\item[\ref{step:neg-forward}']
    If $t$ is in the negative cone $C_i^u$ ($i$ odd) and both successors 
    of $u$ in $C_{i-1}^u$ and $C_{i+1}^u$ are inside $T_{tu}$ (green 
    and red triangle non empty): 
    take one step of side-routing  along the line $sx''$.  
\end{enumerate}

\begin{lemma}\label{l:mem-one-wc}
Constant-memory negative routing has a worst-case routing ratio of  
$\frac{5}{\sqrt{3}}\simeq 2.89$. 
\end{lemma}

\begin{proof}
  Assume without loss of generality $s\in C_0^t$. 
  Referring to the Figure~\ref{fig:lemmas-wc}-right,
  let $x'$ and $x$ be the horizontal and orthogonal projections of 
    $u$ on $T_{st}$, respectively, and $v'$ and $v$ be the horizontal and orthogonal projections of 
    $u$ on $T_{ts}$, respectively.
By Lemma~\ref{f:side-wc-length}, the path from $s$ to $u$ has length 
bounded by  $2\|sx\|$ and 
by Lemma~\ref{f:side-wc-length} again, the path from $u$ to $t$ has length 
bounded by  $2\|vt\|$.
Combining the two paths the length is bounded by 
$2\|sx\|+2\|vt\|
\leq  2\|sx'\|+2\|x'x\|+2\|v't\|
= 2\|wt\| + 2\|xx'\|
$.
Since $x$ is the orthogonal projection of $u$ on the side $x'x''$ of
the equilateral triangle $x'x''v'$, $\|xx'\|$ is smaller than the half side of the
triangle $x'x''v'$ and we get a bound on the path length of
$2\|wt\| + 2\|xx'\|
\leq 2\|wt\| +2\frac 12 \|x'v'\|
\leq 2\|wt\| + \|sw\|$.
Therefore the result follows.
\end{proof}

\section{Probabilistic Analysis}
\label{s:prob}

In this section,
we develop tools to analyze the expected routing ratio of the routing algorithms defined in Sections~\ref{s:basic} and~\ref{s:alternative-neg}
from a probabilistic point of view.
In Section~\ref{ss:forward},
we analyze the expected routing ratio of a forward-routing phase.
In Section~\ref{ss:side},
we analyze the expected routing ratio of a side-routing phase.
Then,
using these results,
we will analyze in Section~\ref{s:wrap} the expected routing ratio of
four different routing algorithms.

\subsection{Routing Ratio of a Forward-Routing Phase}
\label{ss:forward}

Let $X$ be a Poisson point process with intensity $\lambda$
and consider the \halftsix defined on $X\cup\{s\}$,
where $s$ is the origin.
Let $p_0 = s$ and $p_{i+1}$ be the successor 
of $p_i$ in the \halftsix using the cone
$C_0^{p_i}$.
We define  $\tau_1:= \frac{\sqrt3}{12}(3\ln3+4)$.

\newcounter{postponedtheorem}\setcounter{postponedtheorem}{\value{theorem}}
\begin{lemma}\label{l:forward}
Let $A>0$ and  $\alpha=2\sqrt A\lambda^{-\frac14}\sqrt{\log (2A\sqrt\lambda)}$. 
  Consider a forward-routing phase in the upward direction, starting at the origin until 
  it crosses the line $y=A$. 
  The expected routing ratio of this phase is 
  $\tau_1 +O\pth{\frac\alpha A}$. 
  With probability greater than 
  $1- \frac{17}{5} A^{-\frac12}\lambda^{-\frac14}$, 
  the endpoint of this phase lies in
  $[-\alpha, \alpha]\times[A+2\alpha]$. 
\end{lemma}

\begin{proof}
  Let $p_0=s,p_1,p_2,\ldots,p_n$ denote the vertices of the
    forward path.
We consider the following random variables: 
$L_i$ is the Euclidean length of $p_{i-1}p_i$, 
$L_{x,i}$ is the signed length of the horizontal projection of $p_{i-1}p_i$, and 
$L_{y,i}$ is the length of the vertical projection of $p_{i-1}p_i$. 

Since $C_0^{p_i}$ does not intersect $T_{p_{i-1}p_i}$, for
different values of $i$, these variables are independent and identically distributed. 
Thus, 
when there is no ambiguity we denote them by $L$, $L_x$, and $L_y$, 
respectively. 

If $p$ is the upward successor of $s$ in the \halftsix, 
then the upward triangle $\mathcal{T}$ from $s$ having $p$ on its upper boundary is empty. 
Using polar coordinates where $p=r(\cos \phi,\sin\phi)$, 
the area of $\mathcal{T}$, denoted by $\Area{\mathcal{T}}$, is $\frac1{\sqrt3}r^2\sin^2\phi$. 
Thus $\Prob{\mbox{$\mathcal{T}$ is empty}}=e^{-\lambda\Area{\mathcal{T}}}=e^{-\lambda\frac{1}{\sqrt{3}}r^2\sin^2\phi}$, 
from which\footnotemark \footnotetext{Integral computations are
  available as a Maple worksheet with this paper on HAL repository.}
(see Figure~\ref{fig:lemmas-proba-forward}).

\begin{figure}[t]     \begin{center}
         \includegraphics[page=\ipeFigLemmaForward,width=0.33\textwidth]{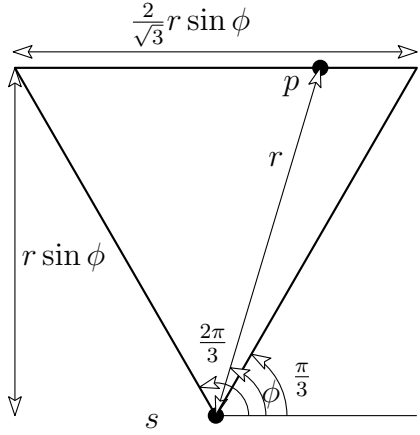}
     \end{center}
     \caption{
For Lemma~\ref{l:forward}. 
        \label{fig:lemmas-proba-forward}
     }
 \end{figure}

\begin{eqnarray*}
\Ex{L}
=\Ex{L_1}
     & = & \lambda \int_{p\in Blue Cone} \Prob{p=p_1} \|p\| \dd p 
\\ & = & \lambda \int_0^\infty \int_{\frac{\pi}{3}}^{\frac{2\pi}{3}}
         e^{-\lambda\frac{1}{\sqrt{3}}r^2\sin^2\phi }r^2 \dd\phi\dd r 
\\ & = & \tfrac{1}{\sqrt{\lambda}} \tfrac{1}{2}3^{-\frac{1}{4}}\sqrt {\pi}(1+\tfrac{3}{4}\ln   3) 
      \simeq \tfrac{1.228}{\sqrt{\lambda}}. 
\end{eqnarray*}

We define $\mu:=\tfrac{1}{2}3^{-\frac{1}{4}}\sqrt{\pi}(1+\tfrac{3}{4}\ln   3) \simeq 1.228$.
By symmetry, 
the expected length of the horizontal projection of an edge is 
$\Ex{L_x}=0$. For the vertical projection, we get 

\begin{eqnarray*}
\Ex{L_y}
  & = &
           \!\lambda \!\!\int_{p\in BlueCone}
           \hspace*{-15mm}
           \Prob{p=p_1} y_p \dd p 
 \quad = \quad 
           \!\lambda \!\!\int_0^\infty \!\!\!\! 
              \int_{\frac{\pi}{3}}^{\frac{2\pi}{3}}  \!\!\!\! 
         e^{-\lambda\frac{1}{\sqrt{3}}r^2\sin^2\!\!\phi } r^2 \sin \phi\  \dd\phi\dd r 
  \\ & = &
           \tfrac{1}{\sqrt{\lambda}} \tfrac{1}{2}3^{\frac{1}{4}}\sqrt {\pi}
      \simeq \tfrac{1.166}{\sqrt{\lambda}} . 
\end{eqnarray*}

We define $\mu_y:=\tfrac{1}{2}3^{\frac{1}{4}}\sqrt {\pi}\simeq 1.166$.
We prove that after $n$ steps, the length of the path from $p_0$ to 
$p_n$ is approximately $n\Ex{L}=\frac{n\mu}{\sqrt{\lambda}}$, while the position of $p_n$
is close to $(n\Ex{L_x}, n\Ex{L_y})=(0,\frac{n\mu_y}{\sqrt{\lambda}})$. 
This gives a routing ratio around 
$\frac{\mu}{\mu_y}$. 
Formally, let us first compute the higher order moments 
and define as follows the constants 
  $\sigma$, 
  $\sigma_x$, 
  $\sigma_y$, 
  $\rho$, 
  $\rho_x$, and 
  $\rho_y$. 
\begin{eqnarray*}
  \Ex{L^2}
  & = &
   \lambda \int_0^\infty \int_{\frac{\pi}{3}}^{\frac{2\pi}{3}}
       e^{-\lambda\frac{1}{\sqrt{3}}r^2\sin^2\phi }r^3 \dd\phi\dd r 
        = 
        \tfrac{10\sqrt{3}}{9\lambda}      = \tfrac{\sigma^2}{{\lambda}}, 
 \\
        \Ex{L_x^2}
   & = &
        \lambda \int_0^\infty \int_{\frac{\pi}{3}}^{\frac{2\pi}{3}}
        e^{-\lambda\frac{1}{\sqrt{3}}r^2\sin^2\phi }r^3 \cos^2\phi\ \dd\phi\dd r 
        = 
        \tfrac{\sqrt{3}}{9\lambda}= \tfrac{\sigma^2_x}{{\lambda}}, 
 \\
        \Ex{L_y^2}
  & = &
        \lambda \int_0^\infty \int_{\frac{\pi}{3}}^{\frac{2\pi}{3}}
        e^{-\lambda\frac{1}{\sqrt{3}}r^2\sin^2\phi }r^3 \sin^2\phi\ \dd\phi\dd r 
        = 
        \tfrac{\sqrt{3}}{\lambda}= \tfrac{\sigma^2_y}{{\lambda}}, 
\\
        \Ex{L^3}
  & = &
        \lambda \int_0^\infty \int_{\frac{\pi}{3}}^{\frac{2\pi}{3}}
        e^{-\lambda\frac{1}{\sqrt{3}}r^2\sin^2\phi }r^4 \dd\phi\dd r 
        = 
        \tfrac{(27\ln 3+68)\sqrt{\pi\sqrt{3}}}{64\lambda\sqrt{\lambda}}
                                       = \tfrac{\rho}{\lambda\sqrt{\lambda}}, 
 \\
        \Ex{|L_x|^3}
  & = &
        2\lambda \int_0^\infty \int_{\frac{\pi}{3}}^{\frac{\pi}{2}}
        e^{-\lambda\frac{1}{\sqrt{3}}r^2\sin^2\phi }r^4 \cos^3\phi\ \dd\phi\dd r 
        = 
        \tfrac{3^{\frac{1}{4}}\sqrt{\pi}}{16\lambda\sqrt{\lambda}}
                                       = \tfrac{\rho_x}{\lambda\sqrt{\lambda}}, 
 \\
        \Ex{|L_y|^3}
  & = &
        \lambda \int_0^\infty \int_{\frac{\pi}{3}}^{\frac{2\pi}{3}}
        e^{-\lambda\frac{1}{\sqrt{3}}r^2\sin^2\phi }r^4 \sin^3\phi\ \dd\phi\dd r 
        = 
        \tfrac{3^{\frac{7}{4}}\sqrt{\pi}}{4\lambda\sqrt{\lambda}}
                                       = \tfrac{\rho_y}{\lambda\sqrt{\lambda}}. 
\end{eqnarray*}

For identically independently distributed variables $X_i$, 
each of which has expected value $\mu_\star$, standard deviation $\sigma_\star$, 
and third moment $\rho_\star$, 
the \emph{central limit theorem}
states that 
\[
\Prob{\absv{\sum_{i=1}^n X_i - n\mu_\star} \geq a\sigma_\star\sqrt{n}}
\leadsto \pth{1- {\rm erf}\pth{ \tfrac{a}{\sqrt{2}} }}, 
\]
when $n$ tends to infinity, 
where \emph{erf} is the error function 
$erf(x):=\frac1{\sqrt\pi}\int_{-x}^x e^{-t^2} dt\in[-1,1]$.
Then, 
the \emph{Berry-Esseen inequality} specifies the rate of convergence: 
\begin{equation}\label{e:Berry-Esseen}
\absv{\rule{0mm}{8mm}
\Prob{\absv{\sum_{i=1}^n X_i - n\mu_\star} \geq a\sigma_\star\sqrt{n}}
- \pth{1-{\rm erf}\pth{ \tfrac{a}{\sqrt{2}} }}
} \leq \frac{\rho_\star}{2\sigma_\star^3\sqrt{n}} . 
\end{equation}

Applying Equation~\eqref{e:Berry-Esseen} to $L_y$ 
with  $a=\sqrt{\log n}$ and using ${\rm erf}(t)\geq 1-\frac{e^{-t^2}}{\sqrt{\pi}t}$ for $t\geq 0$, 
we get 
\begin{eqnarray*}
\Prob{ \absv{ \sum_{i=1}^n L_{y,i} -
    n\tfrac{\mu_y}{\sqrt{\lambda}}} \geq\tfrac{\sigma_y}{\sqrt{\lambda}} \sqrt{n\log n}}
&\leq&  \pth{1-{\rm erf}\pth{ \sqrt{\tfrac{\log n}2} }}
+ \frac{\frac{\rho_y}{\lambda\sqrt{\lambda}}}{2 \pth{\frac{\sigma_y}{\sqrt{\lambda}}}^3\sqrt{n}}
\\&=& \pth{1-{\rm erf}\pth{ \sqrt{\tfrac{\log n}2} }}
+ \frac{\rho_y}{2 \sigma_y^3\sqrt{n}}. 
\\ & \leq&
\pth{1- \pth{1-\frac{\frac1{\sqrt n}}{\sqrt\pi  \sqrt{\frac{\log n}2}}   }}
+ \frac{\rho_y}{2 \sigma_y^3\sqrt{n}}
  \\ & =&
\frac{\sqrt 2}{\sqrt{\pi}\sqrt{n\log n}}
+ \frac{\rho_y}{2 \sigma_y^3\sqrt{n}} . 
\end{eqnarray*}

Let $n_A=\frac{(A+\alpha)\sqrt{\lambda}}{\mu_y}$. 
For $\lambda$ big enough we have 
$\alpha=2\sqrt A\lambda^{-\frac14}\sqrt{\log (2A\sqrt\lambda)}$
smaller than $A$, from which 
\begin{eqnarray*}
  \tfrac{\sigma_y}{\sqrt{\lambda}} \sqrt{n_A\log n_A}
 &=&
  \tfrac{\sigma_y}{\sqrt{\lambda}} \sqrt{\frac{(A+\alpha)\sqrt{\lambda}}{\mu_y}\log \frac{(A+\alpha)\sqrt{\lambda}}{\mu_y}}
\\ &\leq&
  \tfrac{\sigma_y}{\sqrt{\lambda}} \sqrt{\frac{2A\sqrt{\lambda}}{\mu_y}\log \frac{2A\sqrt{\lambda}}{\mu_y}}
  \\ &=&
\sigma_y\sqrt{\tfrac2{\mu_y}}\sqrt A\lambda^{-\frac14}\sqrt{\log (\tfrac2{\mu_y}A\sqrt\lambda)}
\\ &\leq&
3^{\frac14}\sqrt{\tfrac2{\frac123^{\frac14}\sqrt\pi}}\sqrt A\lambda^{-\frac14}\sqrt{\log (2A\sqrt\lambda)}
\\ &\leq&
1.73\sqrt A\lambda^{-\frac14}\sqrt{\log (2A\sqrt\lambda)}
\\ &\leq&
2\sqrt A\lambda^{-\frac14}\sqrt{\log (2A\sqrt\lambda)}
=          \alpha . 
\end{eqnarray*}

\begin{eqnarray*}
  {\Prob{\absv{\sum_{i=1}^{n_A} L_{y,i} - A-\alpha}\geq\alpha }}
  & \leq&
\Prob{ \absv{ \sum_{i=1}^{n_A} L_{y,i} -    A-\alpha} \geq\tfrac{\sigma_y}{\sqrt{\lambda}} \sqrt{n_A\log n_A}}
\\ & \leq&
\frac{\sqrt 2}{\sqrt{\pi}\sqrt{{n_A}\log {n_A}}}
+ \frac{\rho_y}{2 \sigma_y^3\sqrt{{n_A}}}
\\ & \leq&
\tfrac1{\sqrt {n_A}}\pth{\sqrt{\tfrac2\pi}+\frac{\rho_y}{2\sigma_y^3}}
\\ & \leq&
\sqrt{\tfrac{\mu_y}{A\sqrt \lambda}}\pth{\sqrt{\tfrac2\pi}+\frac{\rho_y}{2\sigma_y^3}}. 
\end{eqnarray*}

Now we look at the $x$-coordinate of $p_{n_A}$, 
which is $\sum_{i=1}^{n_A} L_{x,i}$. 
Using Equation~\eqref{e:Berry-Esseen} again, 
we have 

\[
\Prob{\absv{\sum_{i=1}^{n_A} L_{x,i}} \geq \tfrac{\sigma_x}{\sqrt{\lambda}} \sqrt{{n_A}\log {n_A}}}
\leq  \pth{1-{\rm erf}\pth{ \sqrt{\tfrac{\log {n_A}}2} }}
+ \frac{\rho_x}{2 \sigma_x^3\sqrt{{n_A}}}. 
\]

Substituting for ${n_A}$ and since 
$\alpha\geq\tfrac{\sigma_x}{\sqrt{\lambda}} \sqrt{{n_A}\log {n_A}}$
(for $\lambda$ big enough), 
we get 

\[
\Prob{\absv{\sum_{i=1}^{n_A} L_{x,i}} \geq \alpha}
\leq 
\sqrt{\tfrac{\mu_y}{A\sqrt \lambda}}\pth{\sqrt{\tfrac2\pi}+\frac{\rho_x}{2\sigma_x^3}}. 
\]

Thus, 

\begin{eqnarray*}
  \Prob{p_{n_A}\in[-\alpha,\alpha]\times[A,A+2\alpha]}
&  \leq& 
\sqrt{\tfrac{\mu_y}{A\sqrt\lambda}}\pth{2\sqrt{\tfrac2\pi}+\frac{\rho_x}{2\sigma_x^3} +\frac{\rho_y}{2\sigma_y^3}}
\\ &=&
       \sqrt{       \tfrac123^{\frac14}\sqrt\pi}
       \pth{2\sqrt{\tfrac2\pi}+\frac{\frac{3^{\frac14}\sqrt\pi}{16}}{2(3^{-\frac34})^3}
                           +\frac{\frac{3^{\frac74}\sqrt\pi}4}{2(3^{\frac14})^3}}
       A^{-\frac12}\lambda^{-\frac14}
\\ &\leq&
3.4 A^{-\frac12}\lambda^{-\frac14} . 
\end{eqnarray*}

Let $i_A$ be the smallest $i$ such that $y_{p_i}\geq A$. 
If $i_A\leq n_A$, 
then $y_{p_{n_A}}\geq y_{p_{i_A}}\geq A$.
We can bound $\|p_0p_{i_A}\|>A$ and the path length 
$\sum_{i=1}^{i_A} L_i\leq \sum_{i=1}^{n_A} L_i\leq n_A\Ex{L}$. 
If $i_A>n_A$, 
we can 
bound the routing ratio by 
$\frac{2}{\sqrt3}$
using  Lemma~\ref{f:forward-wc-length}
(worst-case analysis of forward-routing). 

\begin{eqnarray*}
  \Ex{\tfrac{\sum_{i=1}^{i_A} L_i}{\|p_0p_{i_A}\|}}
&\leq&
  \Prob{i_A\leq {n_A}} \Ex{\tfrac{\sum_{i=1}^{i_A} L_i}{\|p_0p_{i_A}\|}}
     + \Prob{i_A >  {n_A}} \frac{2}{\sqrt3}
  \\ & \leq &
              1\cdot \frac{{n_A} \cdot \Ex{L}  }A 
              + \sqrt{\tfrac{\mu_y}{A\sqrt \lambda}}\pth{\sqrt{\tfrac2\pi}+\frac{\rho_y}{2\sigma_y^3}} \frac{2}{\sqrt3}
  \\ & \leq &
              \frac{ \frac{\pth{A+\alpha}\sqrt\lambda}{\mu_y} \frac{\mu}{\sqrt\lambda}}A 
              + \sqrt{\tfrac{4\mu_y}3}\pth{ \sqrt{\tfrac2\pi}+\frac{\rho_y}{2\sigma_y^3} }\frac1{A\sqrt\lambda}
  \\ & = &
\frac{\mu}{\mu_y} 
              + O\pth{ A^{-\frac12}\lambda^{-\frac14}\sqrt{\log  (2A\sqrt\lambda)} }
              + O\pth{ \frac1{A\sqrt\lambda} }
  \\ & = &
 \frac{\sqrt3}{12}(3\ln3+4)+O\pth{A^{-\frac12}\lambda^{-\frac14}\sqrt{\log (2A\sqrt\lambda)}}. 
\end{eqnarray*}
\vspace*{-13mm}

\end{proof}

\begin{remark}\label{r:theta6routing1}
  $\Theta_6$-routing in a dense point set runs in two phases, while
  the target is {\em really inside} a cone, it has a similar behavior
  to forward-routing. When it reaches the boundary of a cone, it
  switches between two cones (one odd and one even) and a similar
  analysis can be done. However, measuring the progress along the cone
  boundary instead of the cone bisector gives an expected routing ratio of
 $\tau_2:=\tfrac2{\sqrt3}\tau_1=\tfrac16(3\ln3+4)$.
\end{remark}

\subsection{Routing Ratio of a Side-Routing Phase} 
\label{ss:side}

To analyze the expected routing ratio of a side-routing phase,
we consider the sum of the lengths of all
of its edges. Then, we use the \emph{Slivnyak-Mecke formula}
(refer to~\cite[Corollary 3.2.3]{schneider2008stochastic})
to transform this sum into an integral, from which we get the following lemma.
\begin{lemma}\label{l:side}
  Consider a side-routing phase in the horizontal  direction, starting at the origin until 
  it reaches the line $x=A$. 
  The expected routing ratio of this phase is 
  $ \tau_2+O\pth{A^{-1}\lambda^{-\frac12}}$
  with $ \tau_2:=\frac2{\sqrt3}\tau_1$.
\end{lemma}

\begin{proof}
Up to a scaling factor of $A$ on the lengths and $A^2$ on the density,
we can assume without loss of generality that $A=1$. Let $\nabla(p,q)$ denote the equilateral triangle
with $p$ on its left side, $q$ on its right side and its horizontal
side supported by the $x$-axis (the triangle is below the $x$-axis).
Let $s=(0,0)$ and $t=(1,0)$. The following analysis sums up the lengths of the edges $pq$ of the \halftsix of $X$,
where $p,q\in\nabla(s,t)$ and $\nabla(p,q)$ is empty.
Depending on the situation,
we may have to add
an edge that connects the path to $s$ and/or to $t$ if these points have been added to the
  point set.
We may also have to add an edge that crosses the boundary of $\nabla(s,t)$.
In any case, the expected length of these edges is
  $O\pth{\frac1{\sqrt\lambda}}$, which is negligible.
  Using Slivnyak-Mecke formula, we have
\[ 
\Ex{\mbox{length}} 
 = 
           \Ex{\sum_{p,q\in X\cap\nabla(s,t)}
\hspace*{-4mm} 
           \indicator{\nabla(p,q)\cap X=\emptyset} \|pq\| }
= 
\lambda^2
\int\!\!\!\!\!\int_{p,q\in \nabla(s,t)}  \hspace*{-8mm}  
e^{\lambda \Area{\nabla(p,q)}} \|pq\|
\dd q \dd p.
\] 
To solve this integral, we first define  $\Phi$ the following variable
substitution
(see Figure~\ref{fig:lemmas-proba-side}).
\begin{eqnarray*}
    &&  \Phi : \R\times \R_{\geq 0}\times[0,1]^2 \rightarrow (\R^2)^2  
  \\  
   &&   (x,r,u,v) \leadsto \pth{\;\pth{
     \begin{array}{cc}x-ru\\ -\sqrt{3}r(1-u) \end{array}
  },\pth{
      \begin{array}{cc}x+rv\\ -\sqrt{3}r(1-v)  \end{array}
}}.  
\end{eqnarray*}

\begin{figure}[t]     \begin{center}
         \includegraphics[page=\ipeFigLemmaSide,width=0.43\textwidth]{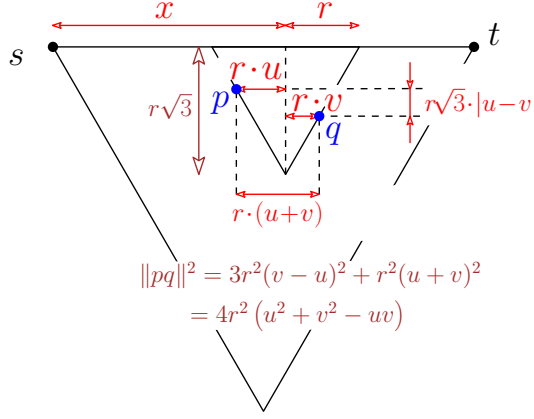}
         \hspace*{0.1\textwidth}
     \end{center}
     \caption{
For Lemma~\ref{l:side}. 
        \label{fig:lemmas-proba-side}
     }
 \end{figure}

 The Jacobian of $\Phi$ is 
\[
{\rm det}(J_\Phi) =
\absv{ \begin{array}{cccc}
                      1 &   -u                     & -r                & 0              \\
                      0 &   -\sqrt{3}(1-u)   & \sqrt{3}r     & 0               \\
                      1 &    v                      & 0                 & r                \\
                      0 &  -\sqrt{3}(1-v)     &0                 & \sqrt{3}r  
      \end{array}
      } = -6r^2.  
\] 

This variable substitution yields 
\begin{eqnarray*}
\Ex{\mbox{length}} 
& = &
\lambda^2 \int_0^1 \int_0^{\min(x,1-x)} \int_0^1 \int_0^1 
e^{-\lambda r^2\sqrt{3}}  
r  2\sqrt{u^2+v^2-uv}
\cdot \absv{ {\rm det}(J_\Phi) }
\dd v \dd u \dd r \dd x 
\\ & = &
\lambda^2 \pth{\int_0^1 \int_0^{\min(x,1-x)} 
e^{-\lambda r^2\sqrt{3}}   
r \cdot 6 \cdot  r^2 
\dd r \dd x 
}\pth{2
\int_0^1 \int_0^1 
\sqrt{u^2+v^2-uv}\, 
\dd v \dd u }
\\ & = &
\lambda^2  \pth{2 \cdot \int_0^{\frac{1}{2}} \int_0^x 
6 e^{-\lambda r^2\sqrt{3}}    
  \cdot  r^3 
\dd r \dd x 
}\pth{4
\int_0^1 \int_0^u 
\sqrt{u^2+v^2-uv}\, 
\dd v \dd u }
\\ & = &
8\lambda^2 
\cdot
 \frac{1}{4\lambda^2}\pth{  
  2
+ 3^{\frac{3}{4}}\sqrt{\pi}{\rm
         \ erf}\pth{\tfrac{1}{2}\sqrt{\lambda}\sqrt[4]{3}}
         \tfrac{1}{\sqrt{\lambda}}
+ e^{-\frac{1}{4}\lambda\sqrt{3}}
}
\pth{\tfrac{1}{6}+\tfrac{\ln 3}{8}}
\\ & = &
\tfrac16(3\ln3+4)+ O\pth{\tfrac{1}{\sqrt{\lambda}}}
=  \tau_2 + O\pth{\tfrac{1}{\sqrt{\lambda}}}
\simeq 1.2160 + O\pth{\tfrac{1}{\sqrt{\lambda}}}.
\end{eqnarray*}
\vspace*{-12mm}

\end{proof}

Observe that side-routing and $\Theta_6$-routing in the
  neighbourhood of a cone boundary have the same expected routing ratio.
This means that in side-routing, the expected progress made by an  edge of the path
  behaves as if it was independent from the previous edges, although this is
  not the case a priori.

\section{Wrap Up}
\label{s:wrap}

In this section,
we analyze the routing algorithms defined in Sections~\ref{s:basic} and~\ref{s:alternative-neg}
using the tools from Section~\ref{s:prob}.
Recall that these algorithms are made of two phases.
The analysis  is based   on the  following two ingredients.
First,
the splitting
 point
between the two phases belongs to a small region of the plane
with high probability.
Second,
the two phases of each routing algorithm
can be analyzed separately.

\begin{figure}[t]     \begin{center}
    \includegraphics[page=\ipeFigCurves,width=0.4\textwidth]{./Figures}\hfill 
     \end{center}
     \caption{
Theorem~\ref{main-theorem}. 
        \label{fig:theorem}
     }
 \end{figure}

\begin{theorem}\label{main-theorem}
Let $X$ be a Poisson point process with intensity $\lambda$.
Consider the positive and the two alternative negative routing algorithms on the \halftsix defined on $X\cup\{s,t\}$,
where $s$ and $t$ are two points at unit distance.
Figure~\ref{fig:theorem}
 presents a graph of the expected routing ratios of the different routing algorithms as $\lambda$ tends
 to $\infty$
 in terms of $\phi$ the angle of $st$ with the horizontal axis.
 The following table shows the expected stretch for a given $\phi$ and
 also the maximum for all $\phi$ and average
 on $\phi$.
\begin{center}\vbox{
\begin{tabular}{|l|c|c|c|}
   \hline
   Routing  &  $\Ex{\mbox{\rm routing ratio}}(\phi) $& $ \max_{s,t}    \Ex{\mbox{\rm routing ratio}} $ & $ \mathbb{E}_{s,t}[\Ex{\mbox{\rm routing ratio}}] $
 \\\hline 
            Positive routing& $\tau_1     \pth{\sin\phi+\tfrac1{\sqrt{3}}\cos\phi}$ &
             $ \tfrac2{\sqrt3}\tau_1\simeq 1.2160$ & $\tfrac{2\sqrt3}\pi\tau_1   \simeq 1.1612 $
   \\\hline 
  Constant-memory& $\tfrac43\tau_1\sin\phi$ &
                   $\tfrac43\tau_1\simeq1.4041$ & $\tfrac4\pi\tau_1   \simeq 1.3408 $
   \\\hline 
  Memoryless\!\! & $      \tau_1 \pth{ \tfrac32 \sin\phi - \tfrac{\sqrt3}6\cos\phi  } $&
                   $\frac32\tau_1\simeq1.5800$ & $\frac{6-\sqrt3}\pi\tau_1 \simeq 1.4306 $
   \\\hline 
 \end{tabular}
\quad $\tau_1:= \frac1{4\sqrt3}(3\ln3+4)$}
\end{center}
\end{theorem}

\begin{proof}
Each of the four routing algorithms we consider is the combination of two phases, each of which can be a forward-routing phase or a side-routing phase.
  These two phases articulate at a splitting point $w\in X$,
  where the
  routing algorithm changes from one phase to the next.
  The actual splitting point is close to an \emph{ideal splitting point} $w_\star$ (which does not belong to $X$)
  that depends only on $s$, $t$ and the routing algorithm.
Therefore, the analysis of each of the four routing algorithms follows the same scheme:
  let the expected routing ratio of the two phases be $\tau$ and $\tau'$, respectively. Since the point $w$ is close to $w_\star$ with high probability,
the expected routing ratio of the routing algorithm is
 $\frac{\|sw_\star\|\tau+\|w_\star t\|\tau'}{\|st\|}$
(see Figure~\ref{fig:proof-theorem}-left). 

\begin{figure}[t]     \begin{center}
\includegraphics[page=\ipeFigProofTh,width=0.99\textwidth]{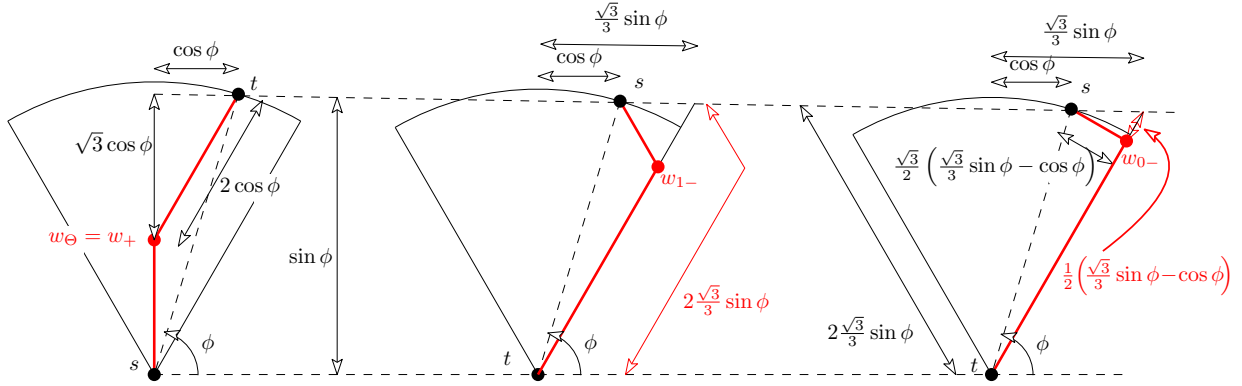} 
\end{center}
     \caption{
For the proof of Theorem~\ref{main-theorem}. 
        \label{fig:proof-theorem}
     }
 \end{figure}

  Positive routing is made of a forward-routing phase from $s$ to $w$ followed by a side 
  phase from $w$ to $t$ (see Section~\ref{s:positive}).
If we fix $s=(0,0)$ and $t=(\cos\phi,\sin\phi)$ for $\phi\in[\frac\pi3,\frac\pi2]$, 
the splitting point 
becomes $w_+=\pth{0,1-\sqrt{3}\sin\phi}$. 
With probability less than $4 \lambda^{-\frac14}$ 
  we bound the routing ratio by the worst case bounds, i.e., between 1
  and 2.
Otherwise, we use the fact that the splitting point $w$ is close to $w_+$.
Let $A=1-\sqrt 3\cos\phi$ and $\alpha=2\sqrt
A\lambda^{-\frac14}\sqrt{\log(2A\sqrt\lambda)}$.
We have
$\|w_+w\|<3\alpha$
with probability greater than  $1- 4 \lambda^{-\frac14}$, by Lemma~\ref{l:forward}.
Moreover, by Lemmas~\ref{l:forward} and~\ref{l:side}, 
$\tau=\tau_1$ and $\tau'=\tau_2$. 
Thus, the expected routing ratio  $\psi:=\Ex{\frac{\tau_1\|sw\|+\tau_2\|w t\|}{\|st\|}} $ satisfies

{\small  \begin{eqnarray*}
\psi           &\leq &
 \Prob{\|w_+w\|\geq\alpha} 2+
\pth{1- \Prob{\|w_+w\|\geq\alpha}}
\pth{ \frac{\tau_1\|sw_+\|+\tau_2\|w_+ t\|}{\|st\|}
 +  \frac{(\tau_1+\tau_2)\|ww_+\|}{\|st\|}}
\\
 \psi          &\geq &
 \Prob{\|w_+w\|\geq\alpha} 1+
\pth{1- \Prob{\|w_+w\|\geq\alpha}}
\pth{ \frac{\tau_1\|sw_+\|+\tau_2\|w_+ t\|}{\|st\|}
 -  \frac{(\tau_1+\tau_2)\|ww_+\|}{\|st\|}}
\\
  \psi         &=&
 \frac{\tau_1\|sw_+\|+\tau_2\|w_+ t\|}{\|st\|}
+ O(\lambda^{-\frac14}\sqrt{\log\lambda})
  \end{eqnarray*}}

And we get as a limit when $\lambda\rightarrow\infty$:
\[
  \frac{\tau_1\|sw_+\|+\tau_2\|w_+ t\|}{\|st\|}  =
  \tau_1 \pth{\sin\phi-\sqrt3\cos\phi} +   \tau_2 2\cos\phi = 
  \tau_1 \pth{\sin\phi+\tfrac1{\sqrt{3}}\cos\phi},  
\]
 whose maximum value is $\tfrac2{\sqrt3}\tau_1\simeq 1.2160$,
obtained for $\phi=\frac{\pi}{3}$. 
Considering any direction equally likely, we average on $\phi$ to get an expected value of 
\[ \frac6\pi\int_{\frac\pi3}^{\frac\pi2}    \tau_1 \pth{\sin\phi+\tfrac1{\sqrt{3}}\cos\phi} \dd\phi
 = \tfrac{2\sqrt3}\pi \tau_1
  \simeq 1.1612 .
\]

\begin{remark}\label{r:theta6routing2}
  By Remark~\ref{r:theta6routing1}, a similar analysis can be done
  for $\Theta_6$-routing with the same 
  splitting point, giving the same expected ratio as with positive
  routing.
\end{remark}



Constant-memory routing is made of a side-routing phase from $s$ to $w$, followed by a side-routing phase from $w$ to $t$ (see Section  \ref{s:mem-one-negative}).
If we fix $t=(0,0)$ and $s=(\cos\phi,\sin\phi)$ for $\phi\in[\frac\pi3,\frac\pi2]$,
the splitting point becomes
$w_{1-}=\pth{\frac{\sqrt3}3\sin\phi-\cos\phi}(1,\frac{\sqrt3}3)$.
By Lemma~\ref{l:side},
$\|w_{1-}w\|<O\pth{\lambda^{-\frac12}}$ with high probability.
Moreover,
by Lemma~\ref{l:side},
 $\tau=\tau'=\tau_2$. 
Thus,
the expected routing ratio when $\lambda\rightarrow\infty$ is 
(see Figure~\ref{fig:proof-theorem}-center). 
\[
  \tau_2 \frac{\|sw_{1-}\|+\|w_{1-} t\|}{\|st\|}=
 \tau_2 2\frac{\sqrt3}3\sin\phi=
   \tfrac43\tau_1\sin\phi,
\]
  whose maximum is $\tfrac43\tau_1\simeq1.4041$,
obtained for $\phi=\frac\pi2$. 
Averaging on $\phi$ yields an expected value of 
\[ \frac6\pi\int_{\frac\pi3}^{\frac\pi2}   \tfrac43\tau_1\sin\phi \dd\phi
= \tfrac{4}{\pi}\tau_1
  \simeq 1.3408.
\]

  Memoryless negative routing is made of a forward-routing phase from $s$ to $w$, followed by a side 
  phase from $w$ to $t$ (see Section~ \ref{s:mem-less-negative}). 
If we fix $t=(0,0)$ and $s=(\cos\phi,\sin\phi)$ for $\phi\in[\frac\pi3,\frac\pi2]$, 
the splitting point 
becomes $w_{0-}$ the orthogonal projection of $s$ on the side of $T_{ts}$
(see Figure~\ref{fig:proof-theorem}-right). 
By Lemma~\ref{l:forward},
$\|w_{0-}w\|<O\pth{\lambda^{-\frac14}\sqrt{\log \lambda}}$
with probability greater than  $1- O( \lambda^{-\frac14})$.
Moreover, by Lemmas~\ref{l:forward} and~\ref{l:side}, $\tau=\tau_1$ and $\tau'=\tau_2$
 Thus,
the limit of the expected routing ratio is 
  {\small
  \begin{eqnarray*}
   \frac{\tau_1\|sw_{0-}\|+\tau_2\|w_{0-} t\|}{\|st\|}
  &=&
      \tau_1 \pth{
      \tfrac{\sqrt3}2\pth{ \frac{\sqrt3}3\sin\phi-\cos\phi}}
      +\tau_2 \pth{
    2\frac{\sqrt3}3\sin\phi-\tfrac12\pth{ \frac{\sqrt3}3\sin\phi-\cos\phi}
}
\\  &=&
      \tau_1 \pth{ \tfrac32 \sin\phi - \tfrac{\sqrt3}6\cos\phi  } ,
\end{eqnarray*}
}

 whose maximum is $\tfrac32\tau_1\simeq1.5800$,
obtained for $\phi=\frac\pi2$. 
Averaging on $\phi$ yields an expected value of 
  \[ \frac6\pi\int_{\frac\pi3}^{\frac\pi2}  
  \tau_1 \pth{ \tfrac32 \sin\phi - \tfrac{\sqrt3}6\cos\phi  }  \dd\phi
  = \tfrac{6-\sqrt3}\pi\tau_1
  \simeq 1.4306  .
\]

Figure~\ref{fig:theorem} 
depicts the expected routing ratios as functions
of $\phi$.
\end{proof}

\subsection*{Acknowledgements}
The authors would like to thank Nicolas Chenavier for interesting discussions related to this paper.

\end{document}